\documentclass[10pt]{article}
\usepackage{amsmath,amssymb,amsthm,amsfonts,fullpage,graphicx,color}
\let\hat\widehat
\let\tilde\widetilde

\usepackage[authoryear,round]{natbib}
\newtheorem{theorem}{Theorem}
\newtheorem{lemma}[theorem]{Lemma}

\DeclareMathOperator*{\argmin}{argmin}

\renewcommand{\P}{\mbox{$\mathbb{P}$}}
\newcommand{\E}{\mbox{$\mathbb{E}$}}

\DeclareSymbolFont{bbold}{U}{bbold}{m}{n}
\DeclareSymbolFontAlphabet{\mathbbold}{bbold}

\newcommand{\one}{\mathbbold{1}}

\newcommand{\Pb}{\mathbb{P}}
\newcommand{\Pn}{\mathbb{P}_n}

\usepackage{tikz}
\usetikzlibrary{positioning,shapes.geometric,calc,shapes,arrows.meta,arrows,decorations.markings, external, trees}
\tikzstyle{Arrow} = [thick,decoration={markings,mark=at position 1 with {\arrow[thick]{latex}}},shorten >= 3pt, preaction = {decorate}]

\begin{document}

\begin{center}
{\bf\Large Conservative Inference for Counterfactuals}\\
Sivaraman Balakrishnan, Edward Kennedy, Larry Wasserman\\
October 17 2023
\end{center}

\begin{quote}
{\em 
In causal inference,
the joint law of a set of
counterfactual random variables
is generally not identified.
We show that a conservative version of the joint law ---
corresponding to the smallest
treatment effect --- is identified.
Finding this law uses recent results
from optimal transport theory.
Under this conservative law
we can bound causal effects and
we may construct inferences
for each individual's counterfactual
dose-response curve.
Intuitively, this is the flattest
counterfactual curve for each subject
that is consistent with the distribution of the observables.
If the outcome is univariate then, under mild conditions, this curve is simply the quantile function
of the counterfactual distribution that passes through the observed point.
This curve corresponds to a nonparametric
rank preserving structural model.
}
\end{quote}

\section{Introduction}

Let $A\in \mathbb{R}$
denote a treatment variable
and $Y$ an outcome of interest.
The counterfactual random variable
$Y(a)$ is
the value $Y$ would have had
if $A$ had been set to $a$.
Suppose we observe
$(A,Y)$ for some subject.
Under the usual assumption of 
no interference,
we have that
$Y = Y(A)$.
Thus, we observe one point of the counterfactual curve
${\cal Y} = (Y(a):\ a\in \mathbb{R})$.
Can we infer the remainder
of the curve
$Y(a)$ for values
$a\neq A$?
In general, the answer is no.
While the marginal distribution
$P_a$ of each counterfactual $Y(a)$
is identified (under standard conditions),
the joint distribution $J$
of the whole curve ${\cal Y}$
is not identified.

The point of this note is to
explain that it is possible
to infer a conservative estimate of
${\cal Y}$.
Specifically, there is a joint distribution $J_*$
--- consistent with the identified marginals $P_a$ ---
that minimizes a certain
treatment effect parameter $\psi$.
Under $J_*$,
the distribution
of $Y(a)$ given that
$Y(A)=Y$,
is degenerate
and puts all its mass
on the curve
\begin{equation}
Y_*(a) = F_a^{-1}(F_A(Y))
\end{equation}
where $F_a(y) = P_a(Y(a) \leq y)$ (and $F_A(Y)$ denotes $F_a(y)$ at the observed $(A,Y)$ values). 
This curve is simply the causal quantile
curve passing through the point
$(A,Y)$.
It is also the optimal transport map
from $Y(A)$ to $Y(a)$.
We will explain why this is true and relate it
to known bounds on causal effects.
The curve $Y_*(a)$
can be seen as a nonparametric rank preserving
structural causal model
\citep{robins1991correcting}.
The connection to optimal
transport inspires us to look at a few
other causal quantities related to
optimal transport.

Figure \ref{fig::cartoon1} shows an example.
There are three counterfactual curves
corresponding to three datapoints
$(A_1,Y_1),(A_2,Y_2),(A_3,Y_3)$.
The points indicate these observations.
Although the curves are not identified,
we can infer the ``flattest'' curves through
the observed points that are consistent with
the marginals which are identified.

Our primary focus is on continuous treatments,
but we also consider discrete treatments as well.
In the binary case
$A\in \{0,1\}$
we can identify
the law $P_0$ of $Y(0)$
and
the law $P_1$ of $Y(1)$
but not the joint law $J_*$
of $(Y(0),Y(1))$.
The optimal transport map $T$
from $P_0$ to $P_1$
defines a joint distribution
$J$ which gives conservative estimates of the joint.
This corresponds to the well-known Frechet-Hoeffding bounds 
\citep{fan2010sharp,heckman1995assessing,aronow2014sharp}.
Figure \ref{fig::cartoon2} shows an example.
The plot shows the density of $P_0$ and $P_1$.
The black dots are the observed $Y_i$'s.
The open red dots show the
imputed values of the other counterfactual value
under the conservative joint law $J$.
Causal estimates computed under $J$
provide sharp lower bounds on various causal effects.

\begin{figure}
\begin{center}
\includegraphics[scale=1]{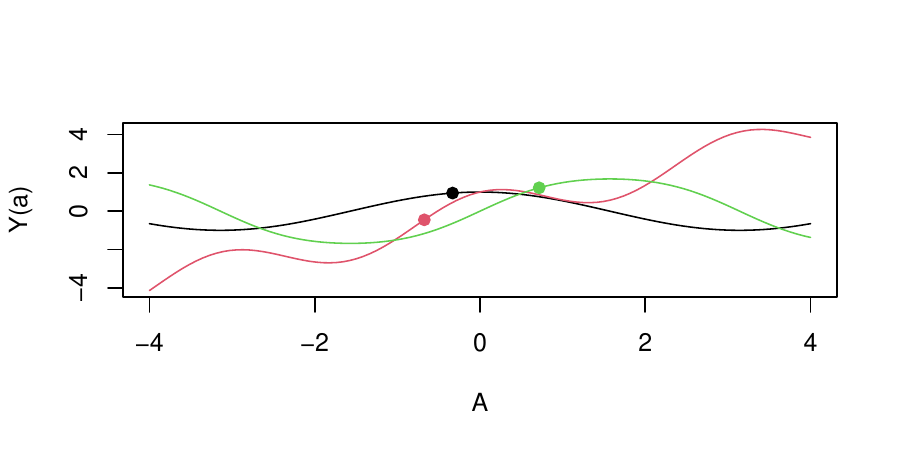}\\
\end{center}
\caption{\em
There are three counterfactual curves
$Y_1(a), Y_2(a), Y_3(a)$.
The points indicate the observations
$(Y_1,A_1),(Y_2,A_2),(Y_3,A_3)$ which correspond
to one point on each curve.
Although the curves are not identified,
we can infer the flattest lines through
the observed points that are consistent with
the marginals which are identified.}
\label{fig::cartoon1}
\end{figure}

\begin{figure}
\begin{center}
\includegraphics[scale=1]{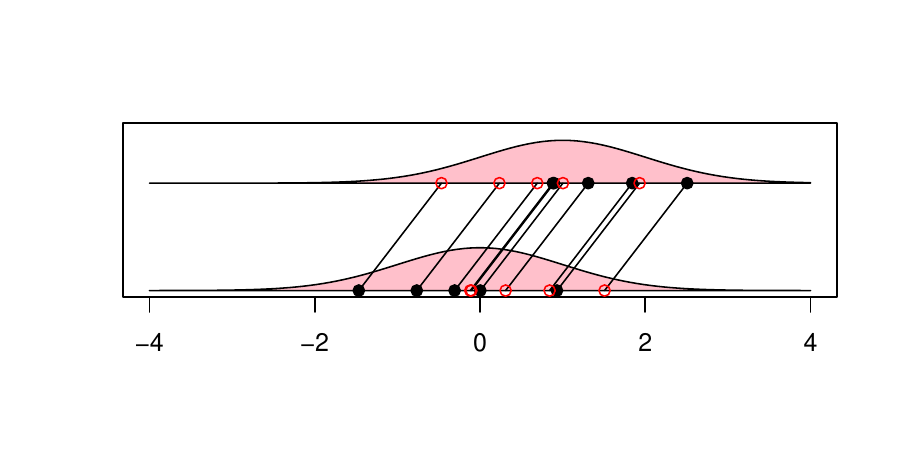}\\
\end{center}
\caption{\em
Binary case.
The black dots are the observed $Y_i$'s.
The open red dots show the
imputed values of the other counterfactual value
under the conservative joint law $J$.}
\label{fig::cartoon2}
\end{figure}

As another illustrative example,
suppose that
in some population there are only
two types of subjects:
those with 
dose response curves of the form $Y(a) = a-1$ 
and those with curves of the form
$Y(a) = -a+1$
with, say, equal probability.
These are the two lines in 
Figure \ref{fig::lines}.
Hence
$\E[Y(a)] = 0$.
The marginal distribution of $Y(a)$ is
$$
F_a(y) = (1/2)\delta_{a-1} + (1/2)\delta_{-a+1}.
$$
If $(A,Y)$ falls on the line
$a-1$ then
$\E_{*}[Y(a)|Y(A)] = a-1$.
If $(A,Y)$ falls on the line
$-a+1$ then
$\E_{*}[Y(a)|Y(A)] = -a+1$.
In this case,
we recover $Y(a)$ perfectly from
$(A,Y)$ but of course this is an extreme example.
A curve that jumps back and forth between
the two lines is consistent with the marginals $F_a$.
But such a curve would have a huge, wild, causal effect.
The flattest curve passing through any one point
is simply the line on which it falls.

\begin{figure}
\begin{center}
\includegraphics[scale=1]{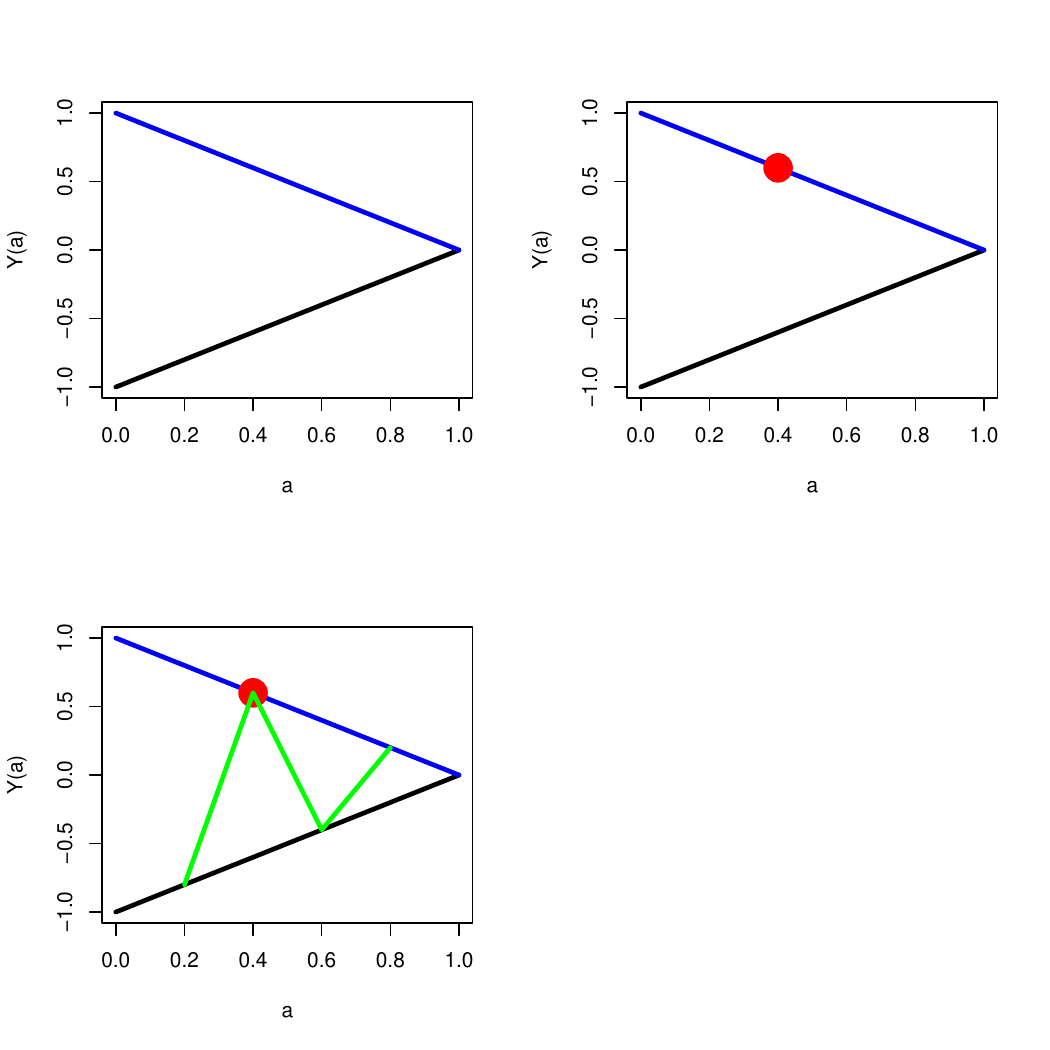}
\end{center}
\caption{\em
In this example,
the curve $Y(a)$ is equal to one of two lines with equal probability.
Given the observation (the large red dot on the top line)
there are infinitely many curves passing through the observation
and that are consistent with the distribution.
One is the blue line itself.
Another is the green line in the bottom left plot.
The most conservative estimate, under the joint distribution that
minimizes $\int\int\E[(Y(a)-Y(b))^2] d\Pi(a) d\Pi(b)$
is the blue line.}
\label{fig::lines}
\end{figure}

We also consider conditional effects.
Specifically, we identify
the law $J_v$
of $Y(a)$ given $V=v$,
that minimizes the causal effect.
And again we can then find the conservative
counterfactual curve
$\E_{J_v}[Y(a)|Y(A)=Y,V=v]$.

In addition to finding conservative inferences,
we also consider estimating the Wasserstein barycenter of the distributions
$P_a$ which provides a summary of these distributions.

{\bf Related Work.}
Finding bounds on causal effects
has a long history.
In the binary treatment case,
the cdf of $Y(1)-Y(0)$
can be bounded using the
Frechet-Hoeffding bound
\citep{heckman1995assessing, fan2010sharp}.
Inference for these bounds was
considered in
\cite{fan2010sharp}.
\cite{manski1997monotone}
found bounds on 
the distribution of monotone treatment effects.
The variance of
$Y(1)-Y(0)$ is studied in
\cite{robins1988confidence,aronow2014sharp}.
As mentioned above, under some conditions,
the conservative estimate of the individual counterfactual dose response curve
is a quantile regression of $Y(a)$.
Causal quantile regression has been studied in 
\cite{firpo2007efficient}.
The problem we study is related to
the problem of inferring the dynamics of cells
given cross-sectional observations of many cells
over time \citep{lavenant2021towards}.
In that case, one has several observations per sample path
and there is no causal model involved.

{\bf Paper Outline.}
We start with the binary case in 
Section \ref{section::binary}.
In Section \ref{section::background}
we give some background on causal inference and optimal transport.
We formally define conservative inference in Section
\ref{section::conservative}.
In Section \ref{section::infinitesimal} we consider the infinitesimal effect
which measures the change in the law of $Y(a)$ under small changes of $a$, which is related to
the conservative effect.
The conditional case is discussed in Section \ref{section::conditional}.
In Section \ref{section::inference} we discuss inference.
Section \ref{section::examples} contains some examples.
Concluding remarks are in Section \ref{section::conclusion}.

\section{Binary Treatments}
\label{section::binary}

Although our goal is to deal with
continuous treatments,
it will be helpful conceptually to begin
with the binary case.
Consider an outcome of interest $Y$,
a treatment $A\in\{0,1\}$ and confounders $X$.
Introduce counterfactual random variables
$Y(0)$ and $Y(1)$
which are the outcomes we would observe if $A=0$ or $A=1$.
The outcome $Y$ is related to the counterfactuals by
$Y = A Y(1) + (1-A) Y(0)$.
Under the conditions in Section \ref{section::background},
$\E[Y(a)]$ is identified and satisfies
$$
\E[Y(a)] = \int \mu(x,a) d\P(x)
$$
where $\mu(x,a) = \E[Y|X=x,A=a]$.
Similarly, the distribution $P_a$ of $Y(a)$
is identified and its cdf is given by
$$
F_a(y) \equiv \P(Y(a) \leq y) = \int \P(Y \leq y|X=x,A=a) d\P(x).
$$
The joint
distribution of
$(Y(0),Y(1))$ is not identified.
Hence, quantities like
$\E[Y(1)|Y(0)]$ are not identified.
But the joint can be bounded by the well-known
Frechet-Hoeffding bounds
\citep{fan2010sharp,aronow2014sharp}.
Specifically,
let $F$ be the joint cdf of
$(Y(0),Y(1))$
Then
$$
L(y_0,y_1) \leq F(y_0,y_1) \leq U(y_0,y_1)
$$
where
\begin{align*}
L(y_0,y_1) &= \max \Biggl\{ F_0(y_0)+ F_1(y_1)-1, 0 \Biggr\},\ \ \ \ \ 
U(y_0,y_1) = \min \Biggl\{ F_0(y_0),F_1(y_1)\Biggr\}.
\end{align*}
From this it may be shown that
$G(t) = P( Y(1)-Y(0) \leq t)$ is bounded by
$G_L(t) \leq G(t) \leq G_U(t)$ where
\begin{align*}
G_L(t) &= \sup_z \max\Biggl\{ F_1(z)-F_0(z-t), 0\Biggr\},\ \ \ \ \ 
G_U(t) = 1+ \inf_z \min\Biggl\{ F_1(z)-F_0(z-t), 0\Biggr\}.
\end{align*}

Now define
$\psi = \E[ (Y(1)-Y(0))^2]$
which is one way of measuring the size of the causal effect.
This parameter is not identified
but can be bounded from the
Frechet-Hoeffding bounds.
Let 
$\mu_0 = \E[Y(0)]$, $\mu_1 = \E[Y(1)]$,
$\sigma_0^2 = {\rm Var}[Y(0)]$ and
$\sigma_1^2 = {\rm Var}[Y(1)]$.
Note that
$$
\E[(Y(1)-Y(0))^2] =
\sigma_0^2 + \sigma_1^2 + (\mu_1-\mu_0)^2 - 2 {\rm Cov}[Y(0),Y(1)]
$$
so lower bounding
$\E[(Y(1)-Y(0))^2]$
is equivalent to upper bounding 
${\rm Cov}[Y(0),Y(1)]$.
We can bound the covariance by
the corresponding Frechet-Hoeffding bound
$$
\int_0^1 F_1^{-1}(u) F_0^{-1}(1-u) du - \mu_0 \mu_1 \leq
{\rm Cov}[Y(0),Y(1)] \leq
\int_0^1 F_1^{-1}(u) F_0^{-1}(u) du - \mu_0 \mu_1.
$$
The upper bound corresponds to the coupling
$(Y(0),Y(1)) = (F_0^{-1}(U),F_1^{-1}(U))$
and the lower bound corresponds to the coupling
$(Y(0),Y(1)) = (F_0^{-1}(1-U),F_1^{-1}(U))$
where
$U\sim {\rm Unif}(0,1)$.
Hence,
\begin{equation}
\psi \geq \E_{J^*}[(Y(1)-Y(0))^2]
\end{equation}
where
$J_*$ is the
joint distribution
of
$(Y(0),Y(1))$
under the minimizing
coupling
$(F_0^{-1}(U),F_1^{-1}(U))$.
Since $J_*$ minimizes the causal effect
we refer to it as the
{\em conservative coupling}.
The lower bound
$\E_{J^*}[(Y(1)-Y(0))^2]$
is identified.
We note that under $J_*$
$$
\bigl(Y(0),Y(1)\bigr) \sim \bigl(Y(0), T(Y(0))\bigr) \sim \bigl( T^{-1}(Y(1)),Y(1)\bigr)
$$
where
$T(y) = F_0^{-1}(F_1(y))$
is the optimal transport map 
(described in the next section)
from
$F_0$ to $F_1$.

Now suppose we observe
$(X,A,Y)$.
What can we say about
$(Y(0),Y(1))$ under the
conservative coupling $J_*$.
Suppose that $A=0$.
Then $Y = Y(0)$ so we know $Y(0)$.
And under $J_*$,
$Y(1) = T(Y(0))$
so we also know $Y(1)$.
More precisely,
the law of
$Y(1)$ given
$(X,A=1,Y)$ is
a point mass at
$T(Y(0))$.
Hence,
$$
\E[ Y(1)| X,A=0,Y] = T(Y(0)).
$$
Similarly,
$$
\E[ Y(0)| X,A=1,Y] = T^{-1}(Y(1)).
$$
More succinctly,
for any $a,a'$,
\begin{equation}
\E[ Y(a)| X,A=a',Y] = T_{a',a} (Y(a'))
\end{equation}
where
$T_{a',a}$ is the optimal transport map from
$F_{a'}$ to $F_a$.
(We can reason similarly, using the coupling that minimizes
$\E[(Y(1)-Y(0))^2|X=x]$;
see Section \ref{section::conditional}.)

Now given $n$ subjects
$(X_1,A_1,Y_1),\ldots, (X_n,A_n,Y_n)$
we have
\begin{equation}
\E[ Y_i(a)| X_i,A_i=a',Y_i] = T_{a',a} (Y_i(a')).
\end{equation}
We can use the $n$ observations to estimate $T_{a,a'}$
and thus we can estimate
$\E[ Y_i(a)| X_i,A_i=a',Y_i]$.
We regard this as our best, conservative estimate
of all the unobserved counterfactuals.
The remainder of this paper
shows how to generalize this
argument to the case 
where the treatment $A$ is a continuous
random variable.

\section{Background}
\label{section::background}

In this section
we review some background material.

\bigskip

{\bf Causal Inference.}
Let $(X,A,Y)\sim P$
where $X\in\mathbb{R}^d$ is a vector of confounding variables,
$A\in\mathbb{R}$ is a treatment and
$Y\in\mathbb{R}$ is an outcome.
Let $\pi(a|x)$ denote the density of $A$ given $X$
which is called the {\em propensity score}.
We make the following
assumptions, which are standard in causal inference:

(A1) No unobserved confounding: $Y(a)$ is independent of $A$ given $X$.\\
(A2) No interference: $Y=Y(A)$.\\
(A3) Overlap: $\pi(a|x) \geq \epsilon > 0$ for all $x$ and $a$.

Under these assumptions,
the distribution $P_a(\cdot)$ of
$Y(a)$ is identified and is given by
\begin{equation}
P_a(B) = \P(Y(a)\in B)=\int P(Y\in B|X=x,A=a) dP(x)
\end{equation}
for every $B$.
We let $F_a$ denote the cdf for $P_a$.
The joint distribution 
of the entire curve ${\cal Y}=(Y(a):\ a\in \mathbb{R})$
is denoted by $J$
but this distribution is, in general, not identified.
Given $n$ observations
$(X_1,A_1,Y_1),\ldots, (X_n,A_n,Y_n)$
there are a variety of estimates of $P_a$.
For example, the plugin estimator is
$\hat P_a(B) = n^{-1}\sum_i \hat P_n(B|X_i,a)$
where $\hat P_n(\cdot |X,A)$ is some estimate of
$P(\cdot|X,A)$;
see Section \ref{section::inference}.
We only observe one point $Y_i = Y_i(A_i)$
on each counterfactual curve
${\cal Y}_i=(Y_i(a):\ a\in\mathbb{R})$.

\bigskip

{\bf Optimal Transport.}
Let $X\sim P$ and $Y\sim Q$.
The joint distribution $J_*$ that minimizes
$\E_J[||X-Y||^2]$
over all joint distributions $J$ having marginals $P$ and $Q$
is the {\em optimal transport plan}
and the value of the minimum
is the squared 2-Wasserstein distance, denoted by
$W^2(P,Q) = \E_{J_*}[||X-Y||^2]$.
If $P$ has a density then $J$ is concentrated
on the graph of a function $T$.
This function, called the {\em optimal transport map},
minimizes $\int ||T(x)-x||^2 dP(x)$
subject to the condition that $T(X)\sim Q$.
If $X$ is real-valued then
$T(x) = G^{-1}(F(x))$
where $F(x) = P(X\leq x)$ and
$G(x) = Q(X\leq x)$.
We note in passing that there is a connection 
between optimal transport maps
and {\em structural nested distribution models}
\citep{robins1989analysis,robins1992estimation,vansteelandt2014structural}.
These 
are models for the distribution of $Y(a)$
of the form
$$
\gamma(y;x,a,\beta) = F_{0|x,a}^{-1}(F_{a|x,a}(y))
$$
where
$F_{b|x,a}(y) = P(Y(b) \leq y|X=x,A=a)$
and $\gamma$ is some function indexed by
the parameter $\beta$.
Thus, we see that
$\gamma(y;x,a,\beta)$ is actually a parametric model
for the optimal transport map from
$Y(a)$ to $Y(0)$,
conditional on $X$ and $A$. \citet{van2003unified} considered parametric 
models for marginalized transport maps, such 
as $F_0^{-1}(F_a(y))$, calling them marginal structural nested distribution models (see for example page 350 of \citet{van2003unified}).

{\bf The Barycenter.}
Given distributions
$(P_a:\ a\in\mathbb{R})$ and a distribution $\nu$ on $A$,
the {\em barycenter} $B$ is the distribution on $Y$ that minimizes
the average Wasserstein distance
$\int  W^2(P_a,B) d\nu(a)$.
In particular we take $\nu = \Pi$ where
$\Pi$ is the marginal distribution of $A$.
We call $B$ the {\em causal barycenter}.
$B$ plays a role in some of the bounds
but it may be of intrinsic interest
as it serves as a summary
of the distributions
$(P_a:\ a\in\mathbb{R})$.
When $Y$ is univariate and $P_a$ has cdf $F_a$,
the barycenter $B$ has cdf $G$
given by
$$
G^{-1}(u) = \int F_a^{-1}(u) d\Pi(a).
$$
If $P_a$ is Gaussian with
mean $\mu(a)$ and covariance $\Sigma(a)$
then $B$ is Gaussian
with mean $\mu = \int \mu(a) d\Pi(a)$
and covariance $\Sigma$ satisfies
$$
\Sigma = \int \sqrt{\Sigma^{1/2}\Sigma(a)\Sigma^{1/2}} d\Pi(a).
$$
Barycenters are useful summaries of sets of distributions
because they are shape preserving.
For example, the barycenter of a
$N(-\mu,1)$ and $N(\mu,1)$ is a $N(0,1)$
whereas the Euclidean average is
the mixture
$(1/2)N(-\mu,1) + (1/2)N(\mu,1)$
which looks nothing like either distribution.

\section{Conservative Causal Effects}
\label{section::conservative}

Let
${\cal J}$ denote all possible joint laws
for ${\cal Y}=(Y(a): a \in \mathbb{R})$
with marginals
$(P_a:\ a\in\mathbb{R})$.
Let $\psi(J)$ be some functional 
that measures a causal effect.
We say that $J_*\in {\cal J}$ is
{\em conservative with respect to $\psi$} if it minimizes
$\psi(J)$.
Examples of such functionals are
the {\em quadratic effect}
$$
\psi(J)= \E_J\Biggl[\int\int  |Y(a)-Y(a')|^2  d\Pi(a) d\Pi(a')\Biggr],
$$
the {\em differential effect}
$$
\psi(J)=  \E_J\Biggl[ \int (Y'(a))^2 d\Pi(a) \Biggr],
$$
and the contrast relative to a baseline $a_0$
$$
\psi(J) =  \E_J \Biggl[\int (Y(a) - Y(a_0))^2  d\Pi(a)\Biggr].
$$
We will primarily focus
on the quadratic effect
although we shall 
discuss the differential effect briefly.

\subsection{Quadratic Effect}

Let $J_*\in {\cal J}$
minimize $\psi(J)$.
We call $J_*$ the
{\em conservative coupling}
as it is the element of ${\cal J}$
that minimizes the given causal effect.
We then define
$\E_{J_*}[Y(a)|X,A,Y]$
which is a conservative guess of the whole 
counterfactual curve
given the observation $(X,A,Y)$.
We will see that, under mild conditions,
the distribution $J_*$ becomes degenerate
once we condition on $(X,A,Y)$.
That is, $J_*$ given $(X,A,Y)$
puts all its mass on a single curve
$Y_*(a)$
given by
$Y_*(a) = F_a^{-1}(F_A(Y))$
which is the quantile curve
passing through the point $(A,Y)$.

Before proceeding
we need one more definition.
Let
$Y(a) = F_a^{-1}(U)$
where $U\sim {\rm Unif}(0,1)$.
Then
$(Y(a):\ a\in \mathbb{R})$ is called the
{\em quantile process} and is contained in ${\cal J}$.
We shall see that, with enough regularity,
the quantile process minimizes several causal parameters.
Going forward, when $A$ is not discrete,
we assume that its distribution $\Pi$ is absolutely continuous with
respect to Lebesgue measure.
We also consider two more assumptions:

{\bf (B1)}: 
Let ${\cal A}$ denote the set of $a$ such that
$P_a$ is absolutely continuous with respect to Lebesgue measure
with bounded density.
Assume that $\Pi({\cal A})>0$.

{\bf (B2)}: 
Each $P_a$ is absolutely continuous with respect to Lebesgue measure
with bounded density.

Clearly (B2) implies (B1).

\bigskip

Assume that $Y\in\mathbb{R}$.
The minimizer of the quadratic effect $\psi$
is given from a result in 
\cite{pass2013optimal}.
Recall 
that the barycenter $B$
is the distribution that minimizes
$V = \int W^2(P_a,B) d\Pi(a)$.
The cdf $G$ corresponding to $B$ satisfies
$$
G^{-1}(u) = \int F_a^{-1}(u) d\Pi(a).
$$
The following lemma
follows directly from
\cite{pass2013optimal}.
He focuses on the case where $\Pi$ is Lebesgue measure on $[0,1]$
but the result holds for arbitrary $\Pi$.

\begin{lemma}
Assume (B1).
Define $J_*$ to be the law of
$(T_{a}(Z):\ a\in\mathbb{R})$
where $Z\sim B$
and $T_a$ is the optimal transport map from $B$ to $P_a$
which necessarily exists.
Then $J_* = \argmin_{J\in {\cal J}}\psi(J)$
and it is the unique minimizer.
The process $J_*$ is continuous in probability:
for every $a$ and every $\epsilon>0$,
$J_*(|Y(b)-Y(a)| > \epsilon)\to 0$ as $b\to a$.
If $P_{a}$ is absolutely continuous then
we may define $J_*$ as
$Y(b) = T^{-1}_b (T_a(Z))$
where $Z\sim P_a$.
Under (B2)
$J_*$ is equal to the quantile process.
\end{lemma}

This process is deterministic once we observe any point on the curve.
Hence we have:

\begin{lemma}
Suppose that $(A,Y)\sim P$
and suppose that (B2) holds.
Then the law of $J_*$ conditional
on $(X,A,Y)$ is degenerate at
$$
Y_*(a) = \E_{J_*}[Y(a)|X,A,Y] = T_a(T_A^{-1}(Y)) = F_a^{-1}(F_A(Y))
$$
where $T_a$ is the optimal transport map from
$Y(a)$ to the barycenter.
\end{lemma}

Let $\tau = F_A(Y)$ so that $Y$ corresponds to the $\tau$ quantile of $P_A$.
Then we see that
$Y_*(a) = q_\tau(a)$, the $\tau$ quantile of $P_a$.
That is
$q_\tau(a) = \inf \{ y:\ F_a(y) \geq \tau\}$ where
$F_a(y) = \int F(y|x,a) dP(x)$.
Hence,
$Y_*(a)$ is the quantile curve passing through $(A,Y)$.

\bigskip

\subsection{The Differential Effect}

Next we consider
the 
{\em differential effect}
$\psi(J) = \E[\int |Y'(a)|^2 d\Pi(a)]$.
We make the following smoothness assumption
throughout this section.
We also assume in this section
that $0 \leq A \leq 1$.

{\bf (AC)}: The map $a\mapsto P_a$ is an absolutely continuous curve in Wasserstein space:
for some $p\geq 2$,
there exists $m\in L^p$ such that
\begin{equation}
W(P_s,P_t) \leq \int_s^t m(r)dr
\end{equation}
for all $s\leq t$.

This assumption leads to many constraints
on the joint distribution $J$.
Thorough accounts are given in
\cite{ambrosio2005gradient,
santambrogio2015optimal,
lisini2007characterization}.
In general, there are many minimizers and
characterizing the set of minimizers is
an open problem \citep{pass2015multi}.
Recently,  
\cite{boubel2021absolutely}
showed the existence of a unique minimizer
under the additional assumption that
$J_*$ is Markov.
Under (B2) and the Markov assumption, $J_*$ is again the quantile process.

Without (B2),
there is still a unique Markov minimizer
which we now describe.
Let
$a_0 < a_1 < \cdots < a_k$
be an ordered subset of values of $A$.
Let
$J_{0,1}$ denote the optimal transport plan
for $P_{a_0}$ and $P_{a_1}$
which, as a joint distribution, thus also defines the conditional 
$J(a_1|a_0)$. 
We define $J(a_2|a_1)$ similarly and so on.
This defines a Markov chain on
$a_0 < a_1 < \cdots < a_k$
which is consistent with the marginals.
We then define the joint process
by taking weak limits as the mesh increases.
This joint law has the following properties
\citep{boubel2018markov,boubel2021absolutely}.
First, $Y(a)$ is Markov:
$$
\text{for all s and all }t>s,\ \ \ \ \ \ 
{\rm Law}(Y(t)| Y(u), u\leq s) = {\rm Law}(Y(t)| Y(s))
$$
and the process is minimal:
$$
\text{Law}(Y(t)| Y(s)\leq y) = \min\Biggl\{ \text{Law}(Z(t)| Z(s)\leq y):\ \text{Law}(Z)\in {\cal J}\Biggr\}
$$
where the minimum is with respect to stochastic ordering, that is,
$\mu \preceq \nu$ if
$\mu(Y \leq y)\geq \nu(Y\leq y)$ for all $y$.
Another way to view this result
is that $J_*$ is the sparsest joint law
subject to the given marginals and smoothness,
where sparsity refers to the
$\ell_0$ norm of the precision matrix
of any finite dimensional marginal.

Condition (AC)
implies other properties as well. 
In particular, the set of distributions
$(P_a:\ 0 \leq a \leq 1)$
satisfies the
{\em continuity equation}
which means that 
\begin{equation}
\partial_a P_a + \partial_y \cdot (v_a P_a) = 0
\end{equation}
for some velocity field $v_a(y)$.
The continuity equation should be interpreted as follows: for every smooth function $f(y,t)$ we have
$$
\int \int 
\Biggl(
\partial_a f(y,a) + \langle v_a(y),\partial_y f(y,t)\rangle 
\Biggr) dP_a(y) da = 0.
$$
It is also equivalent to
$$
\frac{d}{da}\int \xi(y) dP_a(y) = \int \langle \nabla \xi(y),v_a(y)\rangle dP_a(y)
$$
for smooth test functions $\xi(y)$.
We also have that
$\int |v_a(y)|^2 dP_a(y) = D_a^2$
for almost all $a$
where $D_a$ is the Wasserstein derivative defined in (\ref{eq::Da}) in the next section.
Furthermore, almost every sample path satisfies
the differential equation
$Y'(a) = v_a(Y(a))$.
That is, for every $y$, the differential equation
$Y'(a)=v_a(Y(a)), \ \ Y(0)=y$
has a unique solution.
The continuity equation describes the flow of the probability measures $P_a$
while this differential equation
describes the flow of the sample paths.
These are called the Eulerian and Lagrangian views.
Suppose that
the transport map $T_{a,b}$ from $P_a$ to $P_b$ exists.
Define
$$
T'(a) = \lim_{\epsilon\to 0}\frac{T_{a,a+\epsilon}(y) - y}{\epsilon}.
$$
Then
$v_a = T'(a)$. 
(The infinitesimal map $T'(a)$
is similar to the continuous time
structural nested model in
\cite{lok2017mimicking}.)
We conclude that, under the Markov minimizing measure,
$$
\psi = \int \E[(v_a(Y(a)))^2] d\Pi(a) =
\int \int [v_a(F_a^{-1}(u))]^2 du d\Pi(a).
$$

\bigskip

{\bf Summary.}
If we assume that
each $P_a$ is absolutely continuous with respect to Lebesgue measure
then
the quantile process $J_*$ given by
$Y(a) = F_a^{-1}(U)$ with $U\sim {\rm Unif}(0,1)$
is the minimizer of several causal effects.
Given $(X,A,Y)$, this results in the conservative curve
$\E_{J_*}[Y(a)|X,A,Y] = Y_*(a)=T_{A,a}(Y) = F_a^{-1}(F_A(Y))$
where $T_{a,b}$ denotes the optimal
transport map from $Y(a)$ to $Y(b)$.

{\bf Multivariate Treatments.}
The results on the quadratic effect
continue to apply when $A$ is multivariate.
But little is known for other functionals.
It appears that the characterization of the minimizing joint law
is generally an open question.
However, 
once we observe $(X,A,Y)$
we can find 
the conditional minimizer of
$\E[ (Y(a)-Y(A))^2| X,A,Y]$
which,
assuming (B2)
is easily seen to be the law of
$Y_*(a) = T_{A,a}(Y)$
where $T_{a,b}$ denotes
the optimal transport map
from $P_a$ to $P_b$.
This can be regarded as a generalization
of the univariate quantile curve passing through $(A,Y)$.

\bigskip

{\bf Multivariate Outcomes.}
Suppose that $Y\in\mathbb{R}^k$ with $k>1$.
Then, for the quadratic effect,
$$
\psi(J) = \E[ \int\int||Y(a)-Y(b)||^2] d\Pi(a) d\Pi(b) =
\sum_{j=1}^k \E[ \int\int ||Y_j(a)-Y_j(b)||^2] d\Pi(a) d\Pi(b)
$$
so the problem separates into a sum of $k$ pieces.
Any joint $J$ for $Y(a)$ such that
$Y(a) = F_{aj}^{-1}(U_j)$ where $F_{aj}$ is the distribution of $Y_j(a)$ and
$U_j\sim {\rm Unif}(0,1)$, will minimize
$\psi(J)$.
For example, we can take
$Y(a) = F_a^{-1}(U)$ where
$U$ is uniform on $[0,1]^k$
and $F_a$ is the joint cdf of $Y(a)$.
Then
$Y_{j*}(a)\equiv \E_J[Y_j(a)|X,A,Y] = F_{aj}^{-1}(F_{Aj}(Y))$.
Any copula on $[0,1]^k$
will minimize $\psi(J)$ and lead to the same curve $Y_{j*}(a)$.
But the curve
$Y_{*}(a)\equiv \E_J[Y(a)|X,A,Y]$ is not uniquely defined.
It is possible to bound
$Y_{*}(a)$ over the set of all copulas
but we do not pursue that further here.
The differential effect
$\E[\int ||\nabla Y(a)||^2 d\Pi(a)$
can be treated similarly to the quadratic effect.
For the remainder of the paper we confine ourselves to the
univariate case.

\section{The Infinitesimal Effect}
\label{section::infinitesimal}

When $A\in\mathbb{R}$,
a related effect of interest is the
{\em infinitesimal causal effect}
defined by
\begin{equation}\label{eq::Da}
D_a = \lim_{\epsilon\to 0}\frac{W(P_{a+\epsilon},P_a)}{\epsilon}
\end{equation}
where we recall that $W$ is the Wasserstein distance.
This is known as the {\em metric derivative}
and it captures how the distribution $P_a$
changes as $a$ changes.
Note that $D_a$ is a function only of the marginals
and so it is fully identified under the causal assumptions (A1), (A2), (A3), assuming it exists.
Although this is not the same as the counterfactual
quantities we have been considering,
its definition and analysis uses similar tools.

Under (AC),
$D_a$ is guaranteed to exist and
from
\cite{lisini2007characterization}
and 
\cite{stepanov2017three}
we know that
$D_a^2 = \int v_a(y)^2 dP_a(y)$
and there exists 
$J\in {\cal J}$ such that
$D_a^2 = \E_{J}[ Y'(a)^2]$.
In fact,
Theorem 4.1 of 
\cite{stepanov2017three}
shows that, for every $1 < p < \infty$,
there is a $J\in {\cal J}$ such that
$$
D_{a,p} = \int | Y'(a)|^p dJ(Y)
$$
if the curve is sufficiently smooth,
where $D_{a,p}$ is the $p^{\rm th}$ order Wasserstein derivative.
We can also express $D_a$
in terms of
the infinitesimal transport.
That is, 
$D_a^2 = \int [T_a'(y)]^2 dP_a(y)$
where
$$
T_a'(y) = \lim_{\epsilon\to 0}
\frac{T_{a,a+\epsilon}(y) - y}{\epsilon} = 
\lim_{\epsilon\to 0}
\frac{F_{a+\epsilon}^{-1}( F_a(y)) - y}{\epsilon} =
Q_a'(F_a(y))
$$
and 
$Q_a(u) = F_a^{-1}(u)$. 
$T_a'$ is the infinitesimal transport map.

\section{Conditional Effects}
\label{section::conditional}

Let $X= (U,V)$
and define
$P_{a|v}$ to be the distribution of
$Y(a)$ given $V=v$.
Then
$$
P_{a|v}(\,\cdot\,) = \int P(\,\cdot\,|U=u,V=v,A=a) dP(u|v).
$$
In the case $X=V$ we have that
$P_{a|x}(\cdot) = P(\cdot|X=x,A=a)$.
We want to find $J_v$ to minimize
$$
\psi(J_v)=\E\Biggl[ \int\int (Y(a)-Y(a'))^2  \Pi(a)\Pi(a') \Biggm| V=v\Biggr].
$$

As in the unconditional case,
$J_v$ is defined as follows.
Let $B_v$ denote the barycenter of
$(P_{a|v}: 0 \leq a \leq 1)$
which is the distribution that minimizes
$\int W^2(P_{a|v},B_v)d\Pi(a)$.
Let
$Y(a) = T_{a|v}(Z)$
where
$T_{a|v}$ is the optimal transport map from
$B_v$ to $P_{a|v}$.
Specifically,
$T_{a|v}(y) = F_{a|v}^{-1}(G_v(y))$
where $F_{a|v}$ is the cdf of $P_{a|v}$ and
$G^{-1}_v(u) = \int F_{a|v}^{-1}(u) d\Pi(a)$.
Then $J_v$ is the law of this process.
Assuming that $P_{a|v}(\cdot|v)$ is absolutely continuous,
we have
$Y(a) = F_{a|v}^{-1}(\xi|V=v)$ where
$\xi\sim {\rm Unif}(0,1)$.
Then we have
that the corresponding conservative estimate
after observing $(V,A,Y)$ is
\begin{equation}
Y_*(a) = \E[Y(a)|V=v,A,Y] = F_{a|v}^{-1}[F_{A|v}(Y|V=v)].
\end{equation}

\section{Inference}
\label{section::inference}

In this section
we construct estimators
and confidence intervals
for $\psi(J_*)$ and $Y_*(a)$.
We assume that each $F_a$
is absolutely continuous and that its density $p_a$ is bounded.

{\bf The Counterfactual CDF.} 
First we need to 
estimate the cdf
$$ 
F_a(y)  = \Pb(Y(a) \leq y) = \int F(y|x,a) dP(x)
$$
where $F(y|x,a) = P(Y\leq y|X=x,A=a)$.
The efficient influence function for
$F_a(y)$ is
\begin{equation}\label{eq::eifF}
\tilde\varphi(X,A,Y,a,y) =
F(y|X,a) + I(A=a)\frac{ I(Y\leq y)-F_A(y|X)}{\pi(A|X)}.
\end{equation}

If $A$ is discrete
we may use the one step estimator
which is obtained as follows.
Let the sample size be $2n$.
Split the data into two halves.
From the first half, compute an estimator $\hat F(y|x,a)$
of $F(y|x,a)$
and compute an estimate $\hat \pi(a|x)$ of $\pi(a|x)$.
On the second half compute
\begin{equation}
\hat F_a(y) = \frac{1}{n}\sum_i \hat F(y|X_i,a) + \frac{1}{n}\sum_i \hat\varphi_{F_a(y)}(X_i,A_i,Y_i).
\end{equation}
Assuming sufficiently fast rates of convergence for $(\hat{F},\hat\pi)$ (e.g., via sparsity or smoothness) 
this estimator is semiparametric efficient and
$\sqrt{n}(\hat F_a(y) - F_a(y))\rightsquigarrow N(0,\tau^2)$
where
$\tau^2 = \E[\varphi_{F_a(y)}^2(Z)]$
for $Z = (X,A,Y)$.

When $A$ is continuous,
note that $F_a(y)$
takes the form of a dose-response curve but with outcome
equal to $\one(Y \leq y)$. So we can use the DR-learner-type estimator
of \cite{kennedy2017non,bonvini2022fast}.
In this approach, we define
the pseudo-outcome
$$ 
\widehat\varphi_y(Z) = 
\frac{\one(Y\leq y) - \widehat{F}_A(y \mid X)}{\widehat\pi(A \mid X)/\hat\pi(A)} + \Pn\{ \widehat{F}_a(y \mid X)\} |_{a=A} 
$$ 
where 
$Z = (X,A,Y)$ and
$\hat\pi(a)$ is an estimate of the marginal density of $A$. The
estimator is obtained by regressing these pseudo-outcomes on $A$:
$$ 
\widehat{F}_a(y) = \widehat\E_n\{ \widehat\varphi_y(Z) \mid A=a\} 
$$ 
where the notation $\widehat\E_n( \varphi \mid A=a)$ denotes some
generic procedure for regressing $(\varphi_1,...,\varphi_n)$ on
$(A_1,\ldots ,A_n)$ and evaluating at $A=a$. 
Then
\begin{align*}
\E \left\{ \widehat\varphi_y(Z) - \varphi_y(Z) \mid A=a \right\} 
&= \int \left[ \frac{F_a(y \mid X) - \widehat{F}_a(y \mid X)}{\widehat\pi(a \mid X)/\widehat\varpi(a)} \right] d\Pb(X \mid A=a) + 
\Pn\{ \widehat{F}_a(y \mid X)\}  - F_a(y) \\
&=  \int \Big\{ F_a(y \mid X) - \widehat{F}_a(y \mid X) \Big\} 
\left\{ \frac{\pi(a \mid X)/\varpi(a)}{\widehat\pi(a \mid X)/\widehat\varpi(a)} - 1\right\} d\Pb(X) \\
& \hspace{.5in} + (\Pn-\Pb)\{ \widehat{F}_a(y \mid X)\} 
\end{align*}
and hence
we may regress $\hat\varphi_y$ on $A$
to estimate $\hat F_a(y)$.

Alternatively, we can use the plugin estimator
\begin{equation}\label{eq::pi}
\hat F_a(y) = \frac{1}{n}\sum_i \hat F(y|X_i,a)
\end{equation}
where $\hat F(y|x,a)$ is any nonparametric estimator.
For example, we can use the
kernel estimator
$$
\hat F(y|x,a) = 
\frac{\sum_i I(Y_i \leq y) K_h(X_i-x) K_\nu(A_i-a)}
{\sum_i K_h(X_i-x) K_\nu(A_i-a)}
$$
where $K_h$ is a kernel with bandwidth $h$.
Under standard conditions, it is well known that 
$$
\sqrt{n h^{d+1}}(\hat F(y|x,a)- F(y|x,a))\
\rightsquigarrow N(b(x,a,y),\tau^2(x,a,y))
$$
for some bias function $b(x,a,y)$
and some
$\tau(x,a,y)$.
There are various approaches to
eliminating the bias
such as undersmoothing
or using higher order kernels
\citep{mcmurry2004}.
As this is not the focus of this paper
we will ignore the bias
which means that the confidence intervals
we obtain are centered at a smoothed version
of the target.
We also have that
$$
\sqrt{n h}(\hat F_a(y)- F_a(y)\rightsquigarrow N(b(a,y),\tau^2(a,y))
$$
for some $b(a,y)$ and $\tau(a,y)$.
By the functional delta method
$$
\sqrt{n h}(\hat Y_*(a) - Y(a))
\rightsquigarrow N(b(a),v(a))
$$
for some $b(a)$ and $v(a)$.
We avoid estimating $v(a)$
by using the bootstrap.
Specifically, we use the sup norm
bootstrap 
\citep{chernozhukov2014}.
The band is
$\hat Y_*(a) \pm c_\alpha$
where
$c_\alpha$ is the upper $\alpha$ quantile of
$$
\P( |\hat Y_*^*(a) - \hat Y_*(a)| \leq t | Z_1,\ldots, Z_n)
$$
where
$\hat Y_*^*(a)$ denotes a bootstrap draw
give the observed data
$Z_1,\ldots, Z_n$.

If $A$ is discrete,
there are two options for estimating
$T_{a,b}(y)$.
The first is to use the plugin estimator plus the influence function correction:
$$
\hat T_{a,b}(y) = \hat F_b^{-1}(\hat F_A^{-1}(y)) + \frac{1}{n}\sum_i \varphi_{a,b,y}(Z_i)
$$
where the influence function 
$\varphi_{a,b,y}(Z)$ is given below.
Alternatively, we may use
$$
\hat T_{a,b}(y) = \hat F_b^{-1}(\hat F_A^{-1}(y))
$$
where now $\hat F_a$ and $\hat F_b$
are the doubly robust estimators given in
earlier.
These estimators are asymptotically equivalent.

\begin{lemma}
The efficient influence function for
$T_{a,b}(y)$ is
$$
\varphi_{T_{a,b}(y)}(Z) = 
\frac{\varphi_{F_a(y)}(Z) - \varphi_{F_b(y)}(Z) }{p_b(T_{a,b}(Y))}
$$
where $p_b$ is the density of $P_b$.
Let
\begin{equation}
\hat T_{a,b}(y) = \hat F_b^{-1}(\hat F_a(y)).
\end{equation}
\end{lemma}

\begin{proof}
Recall that
$T_{a,b}(y) = F_b^{-1}(F_a(y))$
so that
$$
F_b(T_{a,b}(y)) = F_a(y)
$$
and the chain rule implies that
$$
p_b (T_{a,b}(y)) \dot{T}_{a,b}(y) = \dot{F}_a(y)
$$
where the dot indicates the Gateaux derivative.
The form of $\varphi$ follows
from (\ref{eq::eifF}.
\end{proof}

It follows that,
if $||\hat F(y|a,x)- F(y|a,x)|| = o_P(n^{-1/4})$ and
$||\hat \pi(a|y)-\pi(a|y)|| = o_P(n^{-1/4})$
then
$\hat T_{a,b}(y)$ is semiparametric efficient, 
and
$$
\sqrt{n}(\hat T_{a,b}(y) - T_{a,b}(y))\rightsquigarrow N(0,\tau^2)
$$
where $\tau^2 = \E[ \varphi_{T_{a,b}(y)}^2(Z)]$. This is a nonparametric version of the marginal structural nested model of \citet{van2003unified} (page 350), which may be of independent interest.


.

{\bf The Barycenter}.
Recall that the cdf $G$ for the barycenter $B$ is defined by
$G^{-1}(u) = \int F_a^{-1} (u) d\Pi(a)$
which we can estimate by the plugin estimator
$$
\hat G^{-1}(u) = \frac{1}{n}\sum_i \hat F_{A_i}^{-1}(u).
$$
In the discrete case
we can use the one step estimator
using the following
efficient influence function.

\begin{lemma}
Suppose that $G$ has density $g$
and that $A$ is discrete.
Then
\begin{equation}\label{eq::eifG}
\varphi_{G(y)}(X,A,Y) =
\frac{1}{g(G^{-1}(y))}
\int \frac{\frac{I(A=a)}{\pi(a|X)}(I(Y\leq y)-F(y|X,a)) - F(y|X,a)}{p_a(F_a^{-1}(u))} d\Pi(a)+
\frac{G^{-1}(u) - F_A^{-1}(u)}{g(G^{-1}(y))}.
\end{equation}
\end{lemma}

\begin{proof}
Since $G$ is the inverse of
$\int F_a^{-1}(u) d\Pi(a)$,
by the chain rule,
its influence function 
of $G^{-1}(u)$ is
$$
F_A^{-1}(u) - G^{-1}(u) +
\int M d\Pi(a)
$$
where $M$ is the efficient influence function of
$F_a^{-1} (u)$.
Recall that
the efficient influence function
of $F_a(y)$ is 
$$
\frac{I(A=a)}{\pi(a|X)}(I(Y\leq y)-F(y|X,a)) - F(y|X,a)
$$
and hence the efficient influence function of
$F_a^{-1} (u)$ is 
$$
-\frac{ \frac{I(A=a)}{\pi(a|X)}  (I(Y\leq y)-F(y|X,a)) - F(y|X,a)}
{p_a(F_a^{-1}(u))}
$$
So
the efficient influence function of
$G^{-1} (u)$ is
$$
F_A^{-1}(u) - G^{-1}(u) -
\int \frac{\frac{I(A=a)}{\pi(a|X)}(I(Y\leq y)-F(y|X,a)) - F(y|X,a)}{p_a(F_a^{-1}(u))} d\Pi(a).
$$
This implies that the 
efficient influence function of $G(y)$ is
$$
- \frac{\delta G^{-1} (u)}{g(G^{-1}(y))} =
\frac{1}{g(G^{-1}(y))}
\int \frac{\frac{I(A=a)}{\pi(a|X)}(I(Y\leq y)-F(y|X,a)) - F(y|X,a)}{p_a(F_a^{-1}(u))} d\Pi(a)+
\frac{G^{-1}(u) - F_A^{-1}(u)}{g(G^{-1}(y))}
$$
\end{proof}

In the continuous case
the plugin estimate is
again
$\hat G(y)$ is the inverse of
$$
\hat G^{-1}(u) = \frac{1}{n}\sum_i \hat F_{A_i}^{-1}(u).
$$
There is no efficient influence function
in this case.

\bigskip

{\bf The Quadratic Effect.}
The parameter can be written as
$$
\psi(J_*) = \int \int \int_0^1 (F_a^{-1}(u) - F_b^{-1}(u))^2 \, du \, d\Pi(a) \, d\Pi(b).
$$
The plugin estimator is the $U$-statistic
$$
\hat \psi = \binom{n}{2}^{-1}\sum_{i<j} \int (\hat F_{A_i}^{-1}(u) - \hat F_{A_j}^{-1}(u) )^2 du.
$$
To estimate the parameter efficiently
we use the fact that
$\E[\varphi(Z)] = 0$
where $\varphi$ is the efficient influence function.
This suggests estimating the influence function
and setting its sample mean equal to 0.
Let us start by getting the efficient influence function.

\begin{lemma}
The efficient influence function for
$\psi(J_*)$ is 
\begin{align*}
\varphi(X,A,Y) &=
2\int\int (F_{A}^{-1}(u) - F_a^{-1}(u))^2 \, du \, d\Pi(a) - 2L(X,A,Y) - 2\psi
\end{align*}
where
\begin{align*}
L(X,A,Y) = 
2 \int\int\int (F_a^{-1}(u) - F_b^{-1}(u))^2 
\Biggl(
\frac{\varphi_{F_a(y)}(Z,a,F_a^{-1}(u))}{p_a(F_a^{-1}(u))} - 
\frac{\varphi_{F_b(y)}(Z,b,F_b^{-1}(u))}{p_b(F_b^{-1}(u))} \Biggr) \, du\,d\Pi(a)\,d\Pi(b)\\
\end{align*}
and
$\varphi_{F_a(y)}$ is given in (\ref{eq::eifF}).
\end{lemma}

Since $\E[\varphi(X,A,Y)] = 0$ we see that
$$
\psi = \E[\int\int (F_{A}^{-1}(u) - F_a^{-1}(u))^2 \, du \, d\Pi(a)] - \E[L(X,A,Y)].
$$
Denote the sample size by $2n$.
Split the data into two halves.
From the first half get estimates
$\hat F_a(y)$,
$\hat p_a(y)$ and
$\hat \pi(a|x)$.
Define $\hat\varphi$ by inserting these estimates into the definition of $\varphi$.
We estimate the first term in
$\varphi$ with a $U$-statistic
and the second term with an average.
Setting this equal to 0 gives
the estimator
\begin{align*}
\hat\psi  &=
\binom{n}{2}^{-1}\sum_{i<j} \int (\hat F_{A_i}^{-1}(u) - \hat F_{A_j}^{-1}(u) )^2 du -
\frac{1}{n}\sum_{i} L(X_i,A_i,Y_i)
\end{align*}
where
\begin{align*}
L(X,A,Y) =
\binom{n}{2}^{-1} \sum_{j<k}
\Biggl(
\frac{\tilde\varphi(X,A,Y,A_j,\hat F_{A_j}^{-1}(u))}{\hat p_{A_j}(\hat F_{A_j}^{-1}(u))} - 
\frac{\tilde\varphi(X,A,Y,A_k,\hat F_{A_k}^{-1}(u))}{\hat p_{A_k}(\hat F_{A_k}^{-1}(u))} 
\Biggr)
\end{align*}
and the averages are over the second half of the data.

\begin{theorem}\label{thm::quad}
Supose that
$||\hat F_a(y) - F_a(y)|| = o_P(n^{-1/4})$,
$||\hat p_a(y) - p_a(y)|| = o_P(n^{-1/4})$ and
$||\hat \pi(a|x) - \pi(a|x)|| = o_P(n^{-1/4})$.
Then, if $\psi \neq 0$,
$$
\sqrt{n}(\hat\psi - \psi)\rightsquigarrow N(0,\tau^2)
$$
where
$\tau^2 = \E[\varphi^2]$.
\end{theorem}

Note that $p_a$ is harder to estimate than $F_a$ and so,
$||\hat p_a(y) - p_a(y)|| = o_P(n^{-1/4})$ 
will typically imply that
$||\hat F_a(y) - F_a(y)|| = o_P(n^{-1/4})$.
Furthermore,
$\hat\psi$ is doubly robust, in the sense that
$|\hat\psi-\psi| = O_P(||\hat p_a - p_a||\ ||\hat\pi - \pi||)$.
If $p_a$ and $\pi$ are assumed to belong to Holder spaces,
then the estimator is semiparametric efficient.
Estimating the variance $\tau^2$
is difficult since the estimator is a $U$-statistic.
Instead, we use the HulC \cite{kuchibhotla2021hulc}
which avoids variance estimation.
We divide the data into
$B = \lceil \log(\alpha/2)/\log(2)\rceil$ groups.
We compute estimators
$\hat\psi_1,\ldots,\hat\psi_B$ from each group and set
$C = [\min_j \hat\psi_j,\max_j \hat\psi_j]$.
It follows from
Theorem \ref{thm::quad}
and the results in
\cite{kuchibhotla2021hulc}
that $C$ is a valid, asymptotic $1-\alpha$
confidence interval.
A complication occurs if $\psi\neq 0$.
In this case, the influence function vanishes
and asymptotic Normality fails.
This is a common problem
for quadratic functionals.
For a discussion of this problem
and some possible remedies,
see \cite{verdinelli2023feature}.

\bigskip

{\bf The Differential Effect.}
Under $J_*$ we may
conduct inference for this effect
using its efficient influence function.
This function is very complex
and hence we omit it.

{\bf The Infinitesimal Effect.}
Although not of primary interest for this paper,
let us mention that inference for the
infinitesimal effect
can be obtained using similar tools.
Recall that 
$\psi = \int D^2(a) d\Pi(a)$
where
$$
D(a) = \lim_{\epsilon\to 0} \frac{W(P_{a+\epsilon},P_a)}{\epsilon}.
$$

Now
$$
W^2(P_{a+\epsilon},P_a) =
\int (T_{a,a+\epsilon}(y) - y)^2 dP_a(y) =
\int (F_{a+\epsilon}^{-1}(F_a(y)) - y)^2 dP_a(y)
$$
and
$$
F_{a+\epsilon}^{-1}(u) =
F_a^{-1}(u) + \epsilon \dot{F}_a^{-1}(u) + o(\epsilon) = 
y + \epsilon \dot{F}_a^{-1}(u) + o(\epsilon) 
$$
where
$u = F_a(y)$ and
$$
\dot{F}_a^{-1}(u) = \frac{\partial F_a^{-1}(u)}{\partial a} =
\frac{\dot{F}_a (F_a^{-1}(u))}{p_a (F_a^{-1}(u))}.
$$
So
$$
W^2(P_{a+\epsilon},P_a) = \epsilon^2 
\int \left(\frac{\dot{F}_a (F_a^{-1}(u))}{p_a (F_a^{-1}(u))}\right)^2 dP_a(y) + o(\epsilon^2).
$$
Thus
$$
D^2(a) =  \int \left(\frac{\dot{F}_a (F_a^{-1}(u))}{p_a (F_a^{-1}(u))}\right)^2 dP_a(y).
$$
Estimates can now be plugged into this equation
yielding an estimator
$\hat D^2(a)$.
Efficient estimators of this effect
involve 
the efficient influence function for $\psi$ 
which is very complex.
The details are involved and will be reported
elsewhere.

\section{Example}
\label{section::examples}

We consider an example from
\cite{hirano2004, moodie2012, galagate2015}.
This is a simulated example.
The data are generated as
$$
Y(a) = a + (X_1+X_2)e^{-a(X_1+X_2)} + \epsilon
$$
where $\epsilon\sim N(0,1)$,
$A\sim \exp(X_1+X_2)$,
 and
$X_1,X_2$ are unit exponentials.
It follows that
$$
\E[Y(a)] = a + \frac{2}{(1+a)^3}.
$$
Figure \ref{fig::Example1}
shows 1000 data points drawn from this model.
The black line is an estimate of
$\E[Y(a)]$ due to
\cite{hirano2004}
as implemented in
\cite{galagate2015}.
The other curves show
the estimates of $Y_i(a)$
for 5 randomly selected data points.
The large black dots are
$(A_i,Y_i)$ for these 5 points.
The $x$-axis is on a logarithmic scale.
We estimated the cdf's $F_a$
using kernel regression with a bandwidth of $h=1$.
The second plot shows 90 percent
confidence bands using the bootstrap.
In this example, the uncertainties are very large
except for subject with moderate doses and small values of $Y$.

\begin{figure}
\begin{center}
\begin{tabular}{cc}
\includegraphics[scale=.4]{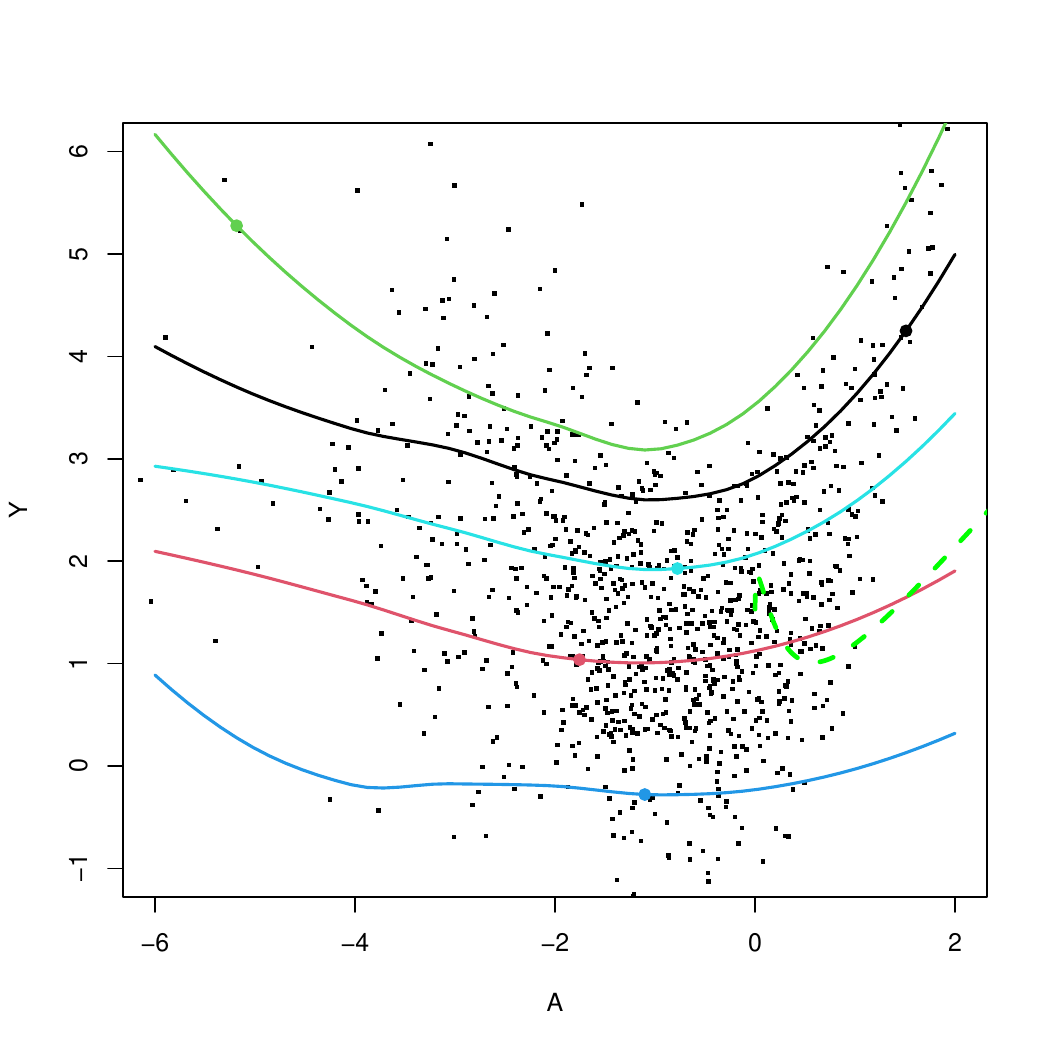}&
\includegraphics[scale=.4]{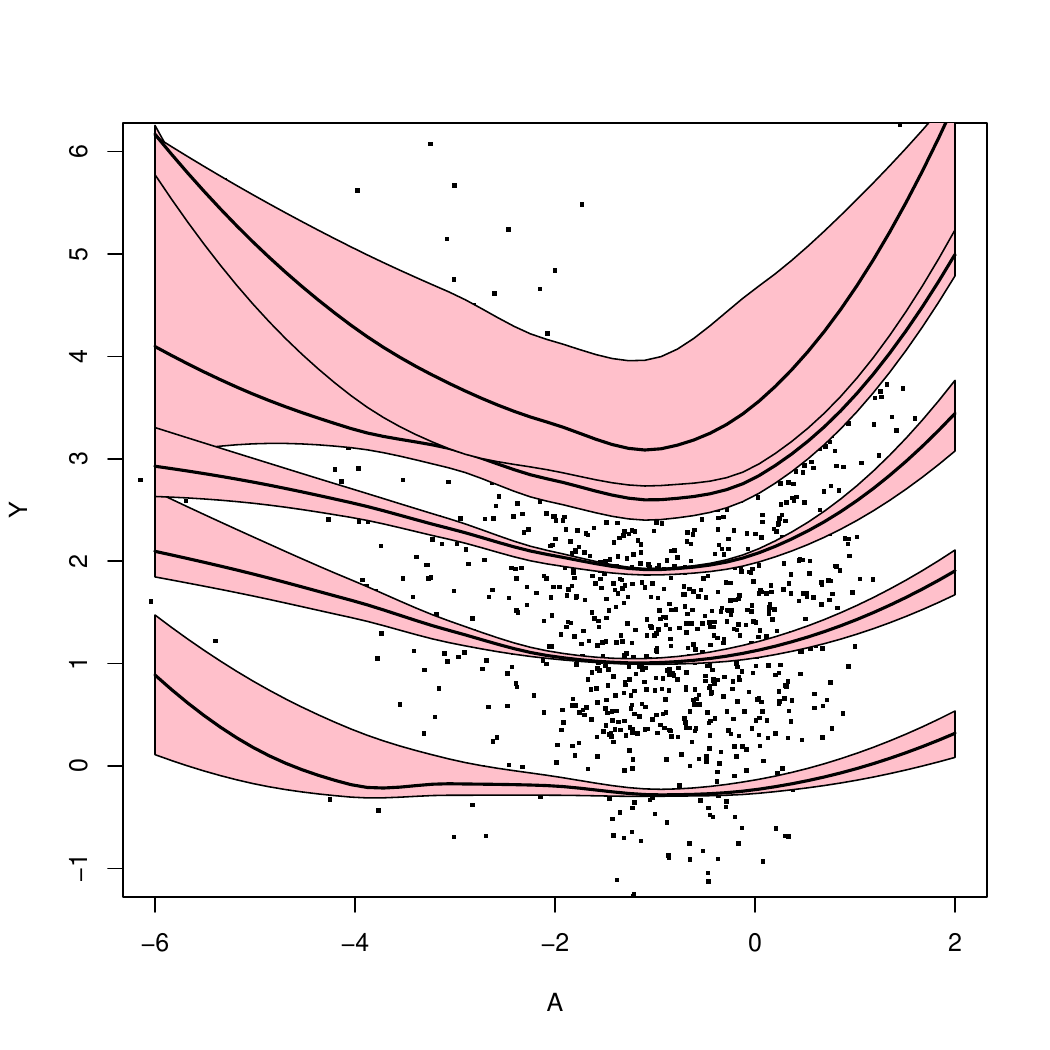}
\end{tabular}
\end{center}
\caption{\em
Left: The data
and estimated counterfactual curves $Y_i(a)$
for a subset of the data.
The dashed line is an estimate of $\E[Y(a)]$.
Right: 90 percent bootstrap confidence bands.}
\label{fig::Example1}
\end{figure}

\section{Conclusion}
\label{section::conclusion}

Although the counterfactual
dose response curve $Y_i(a)$
is not identified,
we can still provide a conservative
estimate of this curve.
Under mild conditions, we have seen
that this amounts to a 
reinterpretation of 
causal quantile estimation.
Equivalently,
$Y_i^*(a)$ is 
is the estimate if $Y(a)$
given $(X,A,Y)$, under
a nonparametric rank preserving
structural nested model.

The analysis uses
tools from optimal transport theory.
The results are a natural generalization
of the Frechet-Hoeffding bounds that
have appeared in the causal literature before.
Once these tools are introduced
we have seen that they suggest
other interesting causal quantities,
such as the causal barycenter
and the infinitesimal Wasserstein effect.

\section*{Appendix: Maximizers}

We may also be interested in 
find the joint measure that maximizes a causal effect.
We focus on binary $A$ and
$\psi(J) = \E[(Y(1)-Y(0))^2]$.
The maximizer of $\psi$
is given by the antitone mapping 
$T(y) = F_1^{-1}[ 1 - F_0(y)]$.
This distributon $J$ is extreme: it says that the treatment
completely reverses the outcome, turning the highest values of $Y(0)$
into the lowest values of $Y(1)$.
In some cases we might want to impose
the monotinicity condition
$Y(1) \geq Y(0)$.
Let ${\cal M}$ denote the
joint distributions in ${\cal J}$
such that $J(Y(1)\geq Y(0)) = 1$.
The joint measure $J\in {\cal M}$ that maximizes
$\psi(J)=\E[ (Y(1)-Y(0))^2]$
was only found recently in 
\cite{nutz2022directional}.
We assume that
$P_0 \preceq P_1$
meaning that
$F_0(y) \geq F_1(y)$
for all $y$.
\cite{nutz2022directional} showed that
$\psi(J)$ is maximized over
${\cal M}$ by the distribution
$J_*$ with cdf
$$
F_*(x,y) = 
\begin{cases}
F_1(y) & y \leq x\\
F_0(x) - \inf_{z\in [x,y]} (F_0(z) - F_1(z)) & y>x.
\end{cases}
$$
The joint $F_*(y(0),y(1))$
is somewhat non-intuitive.
Assuming $P_0$ and $P_1$ are non-atomic
we can describe it is follows.
We have
$$
P_*(x,y) = P_0(x) [ \theta \delta_x(y) + (1-\theta) \delta_{T(x)}(y)]
$$
where
$\delta_x$ is a point mass at $x$,
$\delta_{T(x)}$ is a point mass at $T(x)$,
$$
T(x) = \inf \{ y\geq x:\ F_0(y)-F_1(y) < F_0(x) - F_1(x) \}
$$
and
$$
\theta = \int (p_0(y)\wedge p_1(y)) dy.
$$
It follows that
$$
\E_{J_*}[ Y(1)| Y(0)=Y] = \theta Y + (1-\theta) T(Y).
$$
In these cases,
the mean
$\E_{J_*}[ Y(0)| Y(1)=Y]$
is not a good summary since
the distribution is concentrated on two points.
A better summary is $S(a)$, the support of
$Y(a)$ given $Y(A)=Y$.
\cite{nutz2022directional}
only consider the binary case.
The continuous case
appears to be unexplored.
As in \cite{boubel2021absolutely}
it is likely that the construction can be extended
under a Markov condition
but we do not pursue that here.

\clearpage

\bibliography{paper}

\end{document}